\documentclass[uslettersize, 10pt, conference]{ieeeconf}      

\IEEEoverridecommandlockouts                              
\usepackage{graphics} 
\usepackage{epsfig} 
\usepackage{mathptmx} 
\usepackage{times} 

\usepackage{amsmath,amsthm} 
\usepackage{amssymb}  
\usepackage{amscd}
\usepackage{cite}

\newtheorem{theorem}{Theorem}

\newtheorem{pro}[theorem]{Proposition}
\newtheorem{lem}[theorem]{Lemma}
\newtheorem{cor}[theorem]{Corollary}

\newcommand{\R}{\mathbb{R}}

\title{\LARGE \bf
Formation Control with Triangulated Laman Graphs}

\author{Xudong Chen, M.-A. Belabbas, Tamer Ba\c sar
\thanks{Xudong Chen, M.-A. Belabbas, Tamer Ba\c sar are with the Coordinated Science Laboratory, University of Illinois at Urbana-Champaign, emails: \{xdchen, belabbas, basar1\}@illinois.edu}%
}

\begin{document}

\maketitle
\thispagestyle{empty}
\pagestyle{empty}

\begin{abstract}

Formation control deals with the design of decentralized control laws that stabilize agents at prescribed distances from each other. We call any configuration that satisfies the inter-agent distance conditions a target configuration. It is well known that when the distance conditions are defined via a rigid graph, there is a finite number of target configurations modulo rotations and translations. We can thus recast the objective of formation control as stabilizing one or many of the target configurations. A major issue is that such control laws will also have   equilibria corresponding to configurations which do not meet the desired inter-agent distance conditions; we refer to these as undesired equilibria. The undesired equilibria become problematic if they are also stable. Designing decentralized control laws  whose  stable equilibria are all target configurations in the case of a general rigid graph is still an open problem. We propose here a partial solution to this problem by exhibiting  a  class of rigid graphs and  control laws for which all stable equilibria are  target configurations. 
\end{abstract}

\section{Introduction}

The design of control laws stabilizing a group of mobile autonomous agents has raised a number of issues related to the number and the type of equilibria and their relations to the level of decentralization of the system. In rigidity-based formation control, one assigns agents to the vertices of a  rigid graph 
and specifies the target distances between the pairs of agents linked by  edges. We refer to any configuration of the agents that satisfies these distance requirements as a target configuration. The rigidity of the graph thus ensures that there is a finite number of target configurations up to rotations and translations of the plane. A decentralized formation control law is thus designed to either locally or globally stabilize a subset of the target configurations. However, the decentralization constraints and geometry of the state-space make the appearance of ancillary, undesired configurations inevitable~\cite{AB2013TAC}.
We call a control law {\bf essentially stabilizing} if it only stabilizes  target configurations. 

The relationship between the level of decentralization and the existence of essentially stabilizing control laws has been studied in~\cite{AB2013TAC}, where it was shown that a certain pattern in the information flow of a formation control systems implied the existence of undesired yet stable equilibria. In~\cite{AL2014ECC}, it was shown that one could not locally stabilize all target configurations for a class of directed formations.  Among positive results, it was shown in~\cite{anderson2007control} that the triangle formation was essentially stabilizable and in~\cite{cao2008control} that a class of \emph{acyclic} directed formations was similarly essentially stabilizable. These problems are challenging, and classifying the rigid graphs for which there exists an essentially stabilizing control law is still open. The contribution of this paper is to exhibit a class of \emph{undirected graphs}, termed triangulated Laman graphs, and an associated class of essentially stabilizing control laws for which all stable equilibria are target configurations.

We now describe the model precisely. Let $G=(V,E)$ be an undirected graph with vertex set $V:=\{1,\cdots,N\}$  and edge set $E$. Two vertices are said to be adjacent if there is an edge joining them. We denote by $V_i$ the set of vertices adjacent to vertex $i$.  Let $\vec x_i \in \R^2$, $i=1,\ldots,N$ be the coordinate of agent $i$. With a slight abuse of notation, we will sometimes refer to agent $i$ as agent $\vec x_i$. For every edge $(i,j) \in E$, we let $d_{ij}$ be the distance between agents $i$ and $j$: $$d_{ij}:= \|\vec x_i-\vec x_j\|.$$ We denote by $\bar d_{ij}$ the {\bf target distance} for  $(i,j)\in E$. 

The equations of motion of the $N$ agents $\vec x_1,\cdots,\vec x_N$ in $\mathbb{R}^2$ are given by       
\begin{equation}\label{MODEL}
\dot{\vec x}_{i} = \sum_{j\in V_i}u(d_{ij}, \bar d_{ij})\cdot (\vec x_j-\vec x_i),\hspace{10pt} \forall i= 1,\cdots,N 
\end{equation}
The function $u(d_{ij}, \bar d_{ij})$ is assumed to be jointly continuously differentiable in terms of both arguments. For a fixed $\bar d_{ij}>0$, the function $u(\cdot, \bar d_{ij})$ is monotonically increasing, and it has a unique zero at $\bar d_{ij}$, i.e.,
\begin{equation}
u(d_{ij},\bar d_{ij}) = 0
\end{equation} 
In other words, if all pairs of agents $\vec x_i$ and $\vec x_j$, with $(i,j)\in E$, reach their target distance, then the entire formation is at an equilibrium. For simplicity of exposition, we  assumed in this paper that  the control function $u$ is the same for every pair $(i,j) \in E$. The result however holds for the  general case where different control laws $u_{ij}$'s are used by different pairs of adjacent agents, provided they satisfy the conditions above.

It is known that the dynamics~\eqref{MODEL} is a gradient dynamics (we introduce the potential function in the next section). We can thus rephrase our goal of obtaining an essentially stabilizing control law as designing a potential function whose local minima are all target configurations. 

The undirected formation control model \eqref{MODEL} has been investigated from various perspectives.  Questions concerning the level of interaction laws for organizing such systems~\cite{krick2009, GP, XC2014ACC}, questions about system convergence~\cite{XC2014ACC}, and questions about local stability~\cite{krick2009} and robustness~\cite{AB2012CDC,sun2014CDC,USZB,mou2014CDC} have all been treated to some degree for the case of gradient dynamics. 
Recently, the problem of counting the number of stable equilibria was also addressed in \cite{BDO2014CT,UH2013E}. In general, this is a hard question. For example, even counting the number of equilibria for the gradient formation control system in one dimension is challenging~\cite{BDO2014CT}.

Following this introduction, we proceed as follows. In section II, we  describe  preliminary results about the gradient formation control system. In particular, we will recall some known facts about system convergence and the equivariance of the potential function. In section III, we  introduce the triangulated Laman graph, and then state and prove the main results of this paper. In particular, we introduce in Section~\ref{ssec:morsebott} a formula which can be used to compute the so-called Morse-Bott index of a critical orbit, that allows us to study the type of extremal trajectories of the potential function,  which might be of independent interest. We provide concluding remarks in Section IV, and the paper ends with an Appendix.

\section{Preliminary Results}

\subsection{The control laws and the system convergence}
Let $G=(V,E)$ be an undirected graph with $N$ vertices. We define the  {\it configuration space}  $P_{G}$ of the system as 
\begin{equation}
P_{G}:=\left\{(\vec x_1,\cdots,\vec x_N)\in \mathbb{R}^{2\times N}\big |\vec x_i\neq \vec x_j, \forall (i,j)\in E \right\} 
\end{equation} 
Equivalently, $P_{G}$ is the set of embeddings of the graph $G$ in $\mathbb{R}^2$ whose adjacent vertices have distinct positions. We call a pair $(G,p)$ a {\bf framework}. We  now introduce the class of control laws that is studied in this paper. Let $\mathbb{R}_+$ be the set of positive real numbers, and let $C^1(\mathbb{R}_+,\mathbb{R})$  be the set of continuously differentiable functions from $\mathbb{R}_+$ to $\mathbb{R}$. For fixed $\bar d_{ij}$, the interaction law $u(\cdot, \bar d_{ij})$ can be viewed as an element  in $C^1(\mathbb{R}_+,\mathbb{R})$, and for convenience,  we let 
\begin{equation}\label{fij}
f_{ij}(d) := u(d,\bar d_{ij})
\end{equation}
Denote by {\bf $\mathcal{U}$ the set of functions} $f \in C^1(\mathbb{R}_+,\mathbb{R})$ satisfying the next two conditions:
\begin{itemize}
\item[C1.]  For any  $x>0$, we have  
\begin{equation}
\frac{d}{dx}(xf(x))>0
\end{equation}
and $f(x)$ has a unique zero. 
\item[C2.]  $\lim_{x\to 0}\displaystyle\int^1_{x} t f(t)dt=-\infty$. 
\end{itemize}
The formation control system considered in the paper is then equipped with control laws $u(\cdot,\bar d_{ij}) \in {\cal U}$ for all $\bar d_{ij}$. An example of such a control law is:  
\begin{equation}
u(\|\vec x_i-\vec x_j\|, \bar d_{ij}) =  \frac{\|\vec x_i-\vec x_j\|^2-\bar {d_{ij}}^2}{\|\vec x_i - \vec x_j\|^2}. 
\end{equation}
which is similar to the gradient control law~\cite{krick2009} scaled by $1/\|x_i-x_j\|^2$.

Note that the function $xf(x)$ appears in condition C1 because if $f$ is an interaction law between a pair of agents, then $xf(x)$ represents the actual  attraction/repulsion between them. We impose these two conditions because the first condition implies that the interaction is a monotonically increasing function, so it is a repulsion at a short distance, and an attraction at a long distance. The second condition prevents collisions of adjacent agents along the evolution, so then the solution of system \eqref{MODEL}, with any initial condition in $P_{G}$, exists for all time.    Moreover, we have shown in \cite{XC2014ACC} that if each interaction law $f_{ij}$ satisfies conditions C1 and C2, then all critical orbits of system \eqref{MODEL} are contained in a compact subset of $P_{G}$. Moreover, by assuming $f_{ij}\in  \mathcal{U}$, we have the global convergence of the formation control system \eqref{MODEL} as stated below. 

\begin{lem}\label{CONVGG}
If each $f_{ij}$ is in $\mathcal{U}$, then the set of equilibria of system \eqref{MODEL} is a compact subset of $P_{G}$. Furthermore, for any initial condition $p(0)\in P_{G}$, the solution $p(t)$ of system \eqref{MODEL} converges to the set of equilibria. 
\end{lem}

\subsection{The potential function and its invariance}

An important property of the class of systems~\eqref{MODEL} is that they are gradient flows. The associated potential function is given by 
\begin{equation}\label{PHI}
\Phi(\vec x_1,\cdots,\vec x_N) := \sum_{(i,j)\in E}\displaystyle\int^{d_{ij}}_{1}tf_{ij}(t)dt
\end{equation}
Note that the potential function $\Phi$ depends only on relative distances between agents, thus it is invariant if we translate and/or rotate the entire configuration in $\R^2$. We will now describe this property in precise terms.

The special Euclidean group $SE(2)$ has a natural action  on the configuration space. Recall that  $\gamma$ in $SE(2)$ can be represented by a pair $(\theta, \vec v)$ with $\theta$ a rotation matrix, 
and $\vec v$ a vector in $\mathbb{R}^2$.  With this representation, the  multiplication of two elements $\gamma_1 = (\theta_1,\vec v_1)$ and $\gamma_2 = (\theta_2,\vec v_2)$ of $SE(2)$ is given by
\begin{equation}
\gamma_2\cdot \gamma_1  = (\theta_2 \theta_1, \theta_2 \vec v_1 + \vec v_2)
\end{equation}
The action of $SE(2)$  on $P_{G}$ mentioned above is defined as follows: given  $\gamma = (\theta,\vec v)$ in $SE(2)$ and  $p=(\vec x_1,\cdots, \vec x_N)$ in $P_{G}$ we let
\begin{equation}\label{def:groupaction}
\gamma\cdot p:=(\theta \vec x_1+\vec v,\cdots,\theta \vec x_N+\vec v)
\end{equation} 
We denote by $\mathcal{O}_p$ the orbit of $SE(2)$ through  $p \in P_{G}$: 
\begin{equation}
\mathcal{O}_p:= \{ q \in P_G \mid q= \gamma \cdot p \mbox{ for some } \gamma \in SE(2)\}.
\end{equation} 
The potential function $\Phi$  keeps the same value over $\mathcal{O}_p$:
\begin{equation}
\Phi(p) = \Phi(\gamma\cdot p)
\end{equation} 
for any $p\in P_{G}$ and any $\gamma\in SE(2)$. Denote by $\nabla \Phi$ the gradient of $\Phi$. An immediate consequence of the invariance of $\Phi$ under the group action~\eqref{def:groupaction}  is that
\begin{equation}
\nabla\Phi(\gamma\cdot p) = \operatorname{diag}(\theta,\cdots, \theta) \cdot \nabla\Phi(p)
\end{equation}
where  $\operatorname{diag}(\theta,\cdots,\theta)$ is a block diagonal matrix with $N$ copies of $\theta$.  Since $\operatorname{diag}(\theta,\cdots, \theta)$ is invertible, when $p$ is an equilibrium of system \eqref{MODEL}, then so is $p'$ in $\mathcal{O}_p$.  We  thus refer to the orbit $\mathcal{O}_p$  as a {\it critical orbit} if $\nabla \Phi(p)=0$. 
Let $\mathcal{O}_p$ be a critical orbit, and let $H_p$ be the Hessian matrix of $\Phi$ at $p$, i.e., 
\begin{equation}\label{eqhess}
H_p:= \frac{\partial^2 \Phi(p)}{\partial p^2}
\end{equation}
The following Lemma presents  well-known facts about the Hessian matrix of an invariant function:

\begin{lem}\label{LEQUIV}
Let $\Phi:P \rightarrow \R$ be  a function invariant  under a Lie-group action over a Euclidean space. Denote by $k$ the dimension of a critical orbit $\mathcal{O}_p$ under the group action and denote by  $H_{p}$ be  the Hessian of $\Phi$ at $p$.  Then for any $p_1, p_2 \in \mathcal{O}_p$, the eigenvalues of $H_{p_1}$ and $H_{p_2}$ are the same.  In addition, the Hessian $H_{p}$  has at least $k$ zero eigenvalues. The null space of $H_p$ at least contains the tangent space of $\mathcal{O}_p$ at $p$. 
\end{lem}

In our case, each critical orbit $\mathcal{O}_p$ for $p\in P_{G}$ is of dimension $3$. Let $n_0(H_p)$ be the number of zero eigenvalues of $H_p$. From Lemma \ref{LEQUIV}, we have $n_0(H_p) \ge 3$.  A critical orbit $\mathcal{O}_p$ is said to be {\it nondegenerate} if  $n_0(H_p) = 3$. A potential function $\Phi$ is said to be an {\bf equivariant Morse function} if there are only \emph{finitely} many critical orbits, and moreover each critical orbit  is nondegenerate.

\section{Triangulated Laman Graphs, Independent Partitions and The Morse-Bott Index Formula}
\subsection{Triangulated Laman Graphs}
Let $G=(V,E)$ be an undirected graph. Let the distance function $\rho: P_{G}\to \mathbb{R}^{|E|}_+$ be defined by  
\begin{equation}\label{maprho}
\rho: p \mapsto (\cdots, \|\vec x_i - \vec x_j\|^2,\cdots)
\end{equation}
The graph $G$ is called {\it rigid} in $\mathbb{R}^2$ if for almost all $d \in \R^{|E|}_+$, the pre-image $\rho^{-1}(d)$ is a finite set modulo translations and rotations. 
The graph $G$ is called {\it minimally rigid} if it is not rigid  after taking out any of its edges~\cite{laman1970}. A {\bf Laman graph} is a minimal rigid graph in $\R^2$.
 
It is well known that every Laman graph can be obtained via a so-called \emph{Henneberg sequence}; a Henneberg sequence $\{G_i\}$ is a sequence of minimally rigid graphs obtained via two basic operations: edge split and vertex add. Precisely, start with a graph $G_0$ of only two vertices joined by an edge. Then the graph $G_i$ has $(i+2)$ vertices and is obtained from $G_{i-1}$ by applying one of the two operations. We refer to~\cite{graver1993combinatorial} for more details about these operations. We define {\bf triangulated Laman graphs} as those graphs obtained by imposing constraints on the type of operations allowed: we start with a graph $G_0$ with two vertices connected by one edge. The graph $G_i$ in the sequence is obtained from $G_{i-1}$ by adding  a  vertex and attaching it to two \emph{adjacent} vertices with two new edges. In other words, only the operation of vertex-add is allowed in the Henneberg construction, and in addition, the new vertex cannot be adjacent to two arbitrary vertices, but rather to two vertices connected by an existing edge. See Figure \ref{Hcon} for an illustration.

\begin{figure}[h]
\begin{center}
\includegraphics[scale=.4]{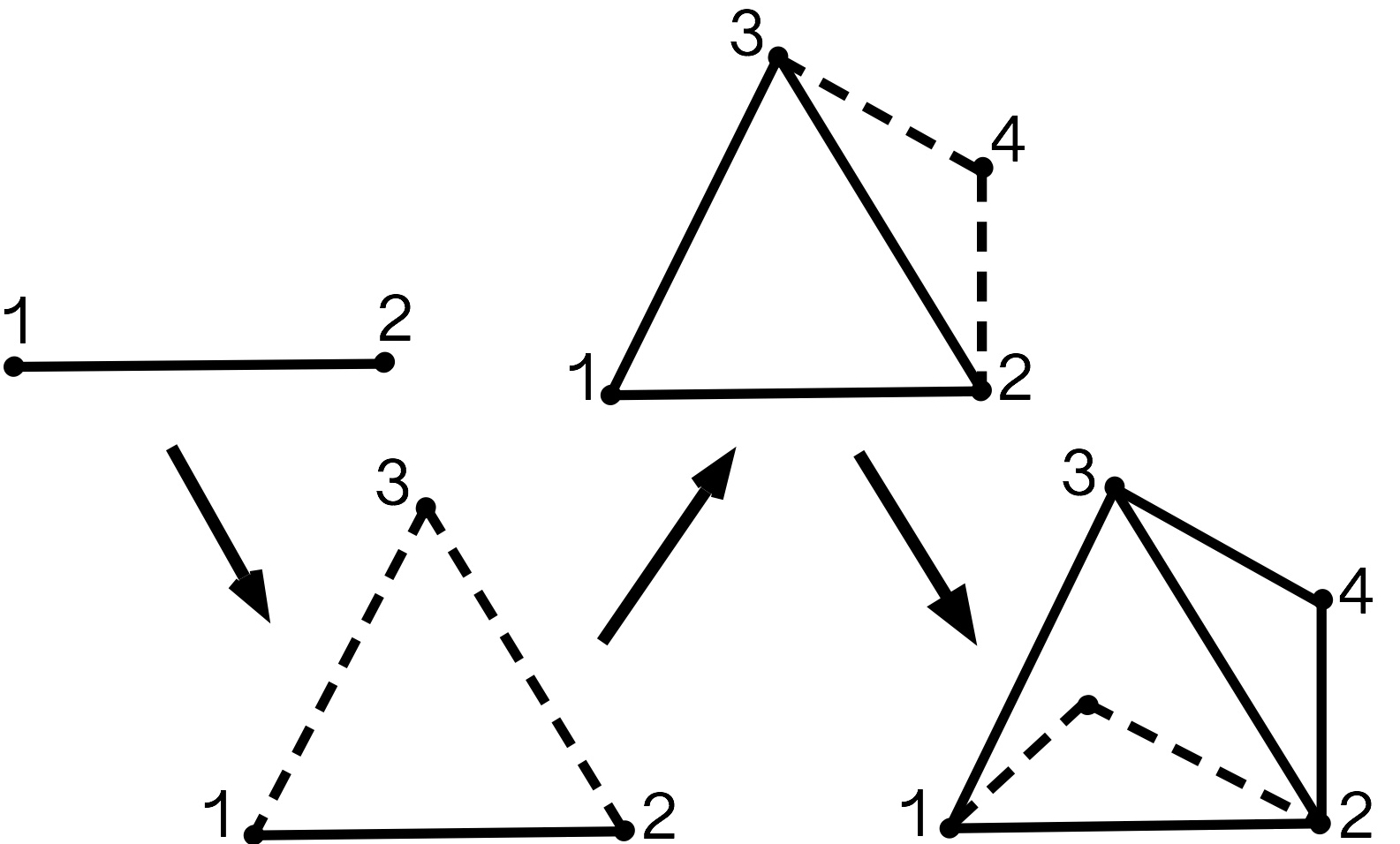}
\caption {An example of a triangulated Laman graph. Start with edge (1,2), then subsequently join vertices 3, 4 and 5 to two existing adjacent vertices.}
\label{Hcon}
\end{center}
\end{figure}

Let $G$ be a  triangulated Laman graph. We say that a subgraph $G'$  of $G$ is a {\it 3-cycle} if $G'$ is a complete graph of three vertices. In graph theory, an {\it induced cycle} of a graph $G$ is a cycle that is an induced subgraph of $G$. If $G$ is a triangulated Laman graph, then all induced cycles of $G$ are the 3-cycles.      
A framework $(G,p)$ is said to be {\it strongly rigid} (or simply $p$ is strongly rigid) if $p$ satisfies the following condition: if vertices $i$, $j$ and $k$ of $G$ form a 3-cycle of $G$, then the triangle formed by agents $\vec x_i$, $\vec x_j$ and $\vec x_k$ is nondegenerate, i.e.,  $\vec x_i$, $\vec x_j$ and $\vec x_k$ are not belong to a one-dimensional subspace of $\mathbb{R}^2$.  If $p$ is strongly rigid, then so is any $p'\in \mathcal{O}_p$.  

 Let $\rho: P_{G}\to \R^{|E|}_+$ be defined by Eq. \eqref{maprho}. A framework $(G,p)$ is said to be {\it infinitesimally rigid}~\cite{graver1993combinatorial} (or simply, $p$ is infinitesimally rigid) if the null space of the Jacobian of $\rho$ at $p$ (i.e., $d\rho(p)/dp$)  is of dimension three. We state below a fact without proof:

\begin{lem}
Strongly rigid configurations are infinitesimally rigid, Moreover, they form an open and dense subset of $P_{G}$.  
\end{lem}

Let $p$ be a strongly rigid configuration, and let $d_{ij}$ be the Euclidean distance between $\vec x_i$ and $\vec x_j$ in $p$. Suppose vertices $i$, $j$ and $k$ form a 3-cycle of $G$, then
\begin{equation}\label{TRINE}
\left\{
\begin{array}{l}
d_{ij} + d_{ik} > d_{jk}\\
d_{ij} + d_{jk} > d_{ik}\\
d_{ik} + d_{jk} > d_{ij}
\end{array}\right.
\end{equation}
We say the set $\{d_{ij}| (i,j)\in E\}$ satisfies the {\bf triangle inequalities} associated with $G$. If the set of desired distances $\{\bar d_{ij} | (i,j)\in E \}$  satisfies the triangle inequalities, then there are strongly rigid configurations satisfying the condition that $d_{ij} = \bar d_{ij}$ for all $(i,j)\in E$; indeed, by following a Henneberg construction, we see that there are $2^{N-2}$ strongly rigid orbits of configuration satisfying these conditions. This exponential relation has also been explored in directed formations~\cite{JB2006cdc}.

We now state in precise terms the main result of this paper.

\begin{theorem}\label{MAIN}
Let $G=(V,E)$ be a triangulated Laman graph and let the target distances $\{\bar d_{ij}|(i,j)\in E\}$  satisfy the triangle inequalities associated with $G$. Let $u(\cdot, \bar d_{ij})\in \mathcal{U}$, for all $(i, j)\in E$, be such that the potential function $\Phi$ defined in~\eqref{PHI} is an equivariant Morse function. 
Then, 
\begin{itemize}
\item[1.] A critical orbit $\mathcal{O}_p$ is (exponentially) stable if and only if it is strongly rigid. There are $2^{N-2}$ stable critical orbits each of which satisfies the condition that $d_{ij}=\bar d_{ij}$ for all $(i,j)\in E$.
\item[2.] For almost all initial conditions $p(0)\in P_{G}$, the solution $p(t)$ of system \eqref{MODEL} converges to one of the $2^{N-2}$ stable critical orbits. 
\end{itemize}
\end{theorem}

The implication of the above is that the control laws considered in this paper are essentially stabilizing the target configurations.




\subsection{Independent Partition}\label{MICF}

We now  introduce the {\bf independent partition} associated with a framework $(G,p)$.  It is a partition of the \emph{edge set} of $G$ such that, roughly speaking, edges that are aligned (with respect to the embedding $p$) are belong to the same subset. Precisely, the independent partition associated with $(G,p)$ can be defined via a Henneberg construction for $G$: given such a Henneberg sequence $\{G'_i\}$, we label the vertices of $G$ with respect to the order in which they appear in the sequence.  
The partition is then constructed in the following way:
\\
\emph{Base case.} Start with the subgraph $G'_0$ of $G$ comprised of vertices $\{1,2\}$. Since there is only one edge $(1,2)$, the partition is trivial.
 \\
\emph{Inductive step.} Now let $G'_i= (V',E')$ be the subgraph of $G$ comprised of vertices $V'= \{1,\cdots,i+2\}$ and  assume that we have partitioned $E'$ into disjoint subsets as 

\begin{equation*}
E' = E'_1\cup\cdots \cup E'_{m'}
\end{equation*} 
Suppose that in the chosen Henneberg construction, vertex $(i+3)$ links to vertices $j$ and $k$ via edges $(j,i+3)$ and $(k,i+3)$. Without loss of generality, we assume that $(j,k)\in E'_1$. Then we consider two cases:
\begin{enumerate}
\item[] {\it Case I}. If $\vec x_{j}$, $\vec x_k$ and $\vec x_{i+3}$ are aligned, then update the partition by adding $(j,i+3)$ and $(k,i+3)$ into $E'_1$.
\vspace{3pt}

\item[] {\it Case II}. If $\vec x_i$, $\vec x_j$ and $\vec x_{i+3}$ are not aligned, then update the partition  as
\begin{equation*}
E'_1\cup\cdots\cup E'_{m'}\cup\{(j,i+3)\}\cup\{(k,i+3)\}
\end{equation*}

\end{enumerate}
By following the Henneberg construction, we then derive the independent partition of $E$ associated with $(G,p)$. We note that the independent partition does not rely on the choice of the Henneberg construction~\cite{chenRMAP1}. We refer to Fig.~\ref{DECO} for an illustration.

\begin{figure}[h]
\begin{center}
\includegraphics[scale=.4]{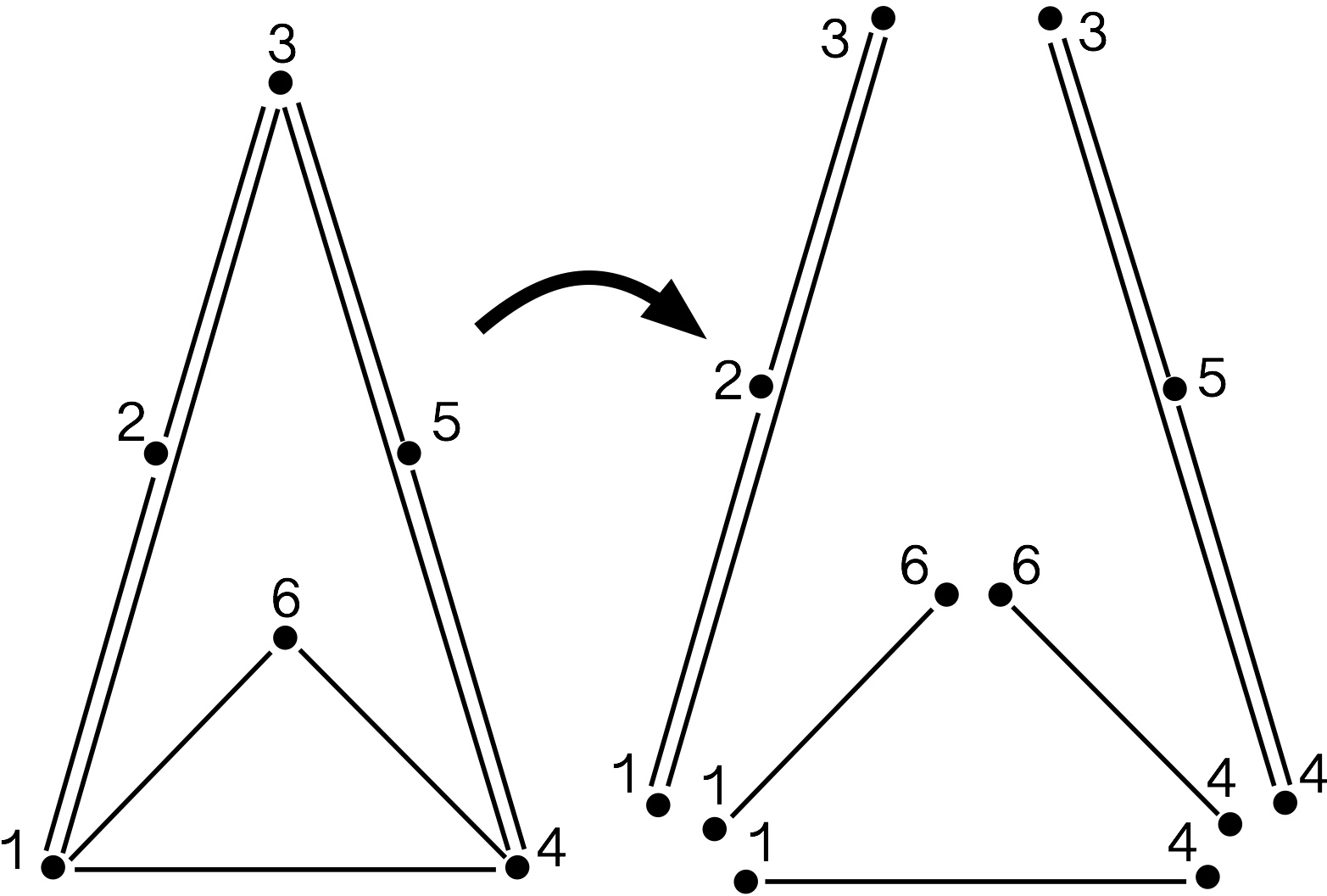}
\caption {An example of the independent partition. We see from the left figure that the graph $G$ is a triangulated Laman graph as we label the vertices with respect to a Henneberg construction, and $p$ is a planar configuration with $\vec x_1, \vec x_2,\vec x_3$ aligned, and $\vec x_3,\vec x_4,\vec x_5$ aligned. Then the independent partition of $E$ associated with $(G,p)$ is given by the right figure.}
\label{DECO}
\end{center}
\end{figure}

Let $\{E_1,\cdots, E_m\}$ be the disjoint subsets of edges associated with the independent partition for $(G, p)$. Let $V_i$ be the set of vertices incident to edges in $E_i$,  let $G_i:=(V_i,E_i)$, and let $(G_i,p_i)$ be the corresponding framework. We  summarize some  properties associated with independent partitions.

\begin{pro}\label{indpar}
Let $\{(G_i,p_i)\}$ be the frameworks associated with the independent partition for $(G,p)$. Then
\begin{enumerate}
\item\label{cond1}  Each $G_i$ is a triangulated Laman graph.
\vspace{2pt}

\item\label{cond2}  Each $(G_i,p_i)$ is a line framework. 
\vspace{2pt}

\item If there is another partition of $E$ satisfying  conditions 1) and 2), then it is a refinement of the independent partition.  In other words, the independent partition contains minimal number of subgraphs  satisfying  conditions~\ref{cond1} and~\ref{cond2}.
\vspace{2pt}

\item  If in addition $p$ is an equilibrium of system~\eqref{MODEL}, then each $p_i$ is an equilibrium of the subsystem induced by $G_i$. 
\vspace{2pt}

\item If in addition $p$ is strongly rigid, then each $p_i$ is a configuration of two agents, i.e., the edge set $E_i$ of $G_i$ is a singleton.

\end{enumerate}
\end{pro}

More details, including proofs of these statements, can be found in \cite{chenRMAP1}.

\subsection{The Morse-Bott Index Formula}\label{ssec:morsebott}
Let $\mathcal{O}_p$ be a critical orbit of system \eqref{MODEL}.  
Let $n_+(H_p)$, $n_0(H_p)$, and $n_-(H_p)$ be the numbers of positive, zero, and negative eigenvalues of $H_p$, respectively. 
We refer to the triplet  
\begin{equation}
\vec n(H_p) = (n_+(H_p), n_-(H_p), n_0(H_p))
\end{equation} 
as the  {\it signature} of  $H_p$. By Lemma \ref{LEQUIV}, the signature of $H_{p'}$ is invariant as $p'$ varies over $\mathcal{O}_p$. Note that in terms of the signature,  a critical orbit $\mathcal{O}_p$ is {\it exponentially stable} if and only if 
\begin{equation}
\vec n(H_p) = (0, 2N -3,  3)
\end{equation} 
Define the {\bf Morse-Bott index} and {\bf co-index} of $\mathcal{O}_p$ to be $n_-(H_p)$ and $n_+(H_p)$ respectively. We now show how to evaluate these two indices of a critical orbit.

Let $G' = (V',E')$ be a subgraph of $G$. A formation control system is said to be {\it induced by} $G'$  if it is comprised of  agents $\vec x_i$ for $i\in V'$  together with  $f_{ij}$'s the interaction laws for $(i,j)\in E'$. To be precise,   the equations of motion for the subsystem induced by $G'$ are
\begin{equation}
\dot{\vec x}_{i} = \sum_{j\in V'_i}u(d_{ij}, \bar d_{ij})\cdot (\vec x_j-\vec x_i),\hspace{10pt} \forall i\in V'
\end{equation}
with $V'_i$ the neighbors of $i$ in $G'$. 
The subsystem is  a gradient flow  for the  potential function
 \begin{equation}
\Phi'(p') :=  \sum_{(i,j)\in E'} \displaystyle \int^{d_{ij}}_1 tf_{ij}(t) dt. 
\end{equation}
with $f_{ij}$ defined in~\eqref{fij}.

\begin{pro}\label{MBIF}
Let $G$ be a triangulated Laman graph.  Let $p$ be an equilibrium of system \eqref{MODEL}, and let $\{(G_i,p_i)\}^m_{i=1}$  be the frameworks associated with the independent partition for $(G,p)$. Let 
$
\Phi_{i} 
$
be the potential function of the subsystem induced by $G_i$. 
Let $H_{p_i}$ be the Hessian of $\Phi_i$ at $p_i$. Then
\begin{equation}\label{INDEXF}
\left\{
\begin{array}{l}
n_-(H_p)=\sum^m_{i=1}n_-(H_{p_i}) \vspace{3pt}\\
n_+(H_p)=\sum^m_{i=1}n_+(H_{p_i})
\end{array}
\right.
\end{equation}
This set of expressions will be referred as the {\bf Morse-Bott index formula}. 
\end{pro}

We provide a sketch of the proof of Proposition \ref{MBIF} in the Appendix, and we refer to \cite{chenRMAP1} for a complete proof. 
Proposition~\ref{MBIF} is used to prove the following Corollary.

\begin{cor}\label{cor6}
The critical orbit $\mathcal{O}_p$ is nondegenerate if and only if each $\mathcal{O}_{p_i}$ is nondegenerate. Moreover, 
the critical orbit $\mathcal{O}_p$ is exponentially stable if and only if each $\mathcal{O}_{p_i}$ is exponentially stable. 
\end{cor}

\begin{proof}
Let $|V_i|$ and $|E_i|$ be the cardinalities of $V_i$ and  $E_i$, respectively. Since each $G_i$ is a triangulated Laman graph, 
\begin{equation}
|E_i| = 2|V_i| - 3.
\end{equation} 
By Lemma \ref{LEQUIV},  we have $n_0(H_{p_i}) \ge 3$, and hence $n_+(H_{p_i}) + n_-(H_{p_i}) \le |E_i|$. On the other hand, we have  
\begin{equation}
|E| = \sum^m_{i=1} |E_i|
\end{equation} 
Thus, by Proposition \ref{MBIF}, we know that 
\begin{equation}
\left\{
\begin{array}{ll}
n_+(H_p) + n_-(H_p) \le |E| \vspace{3pt}\\
n_0(H_p) \ge 3
\end{array}
\right.
\end{equation}
The equalities hold if and only if $n_0(H_{p_i}) = 3$ for all $i$. Thus, the critical orbit $\mathcal{O}_p$ is nondegenerate if and only if each $\mathcal{O}_{p_i}$ is nondegenerate. Also by Proposition \ref{MBIF}, $n_+(H_p) = 0$ if, and only if, $n_+(H_{p_i}) = 0$ for all $i$. This competes the proof.   
\end{proof}

From Proposition \ref{MBIF} and Corollary \ref{cor6}, we see that it suffices to understand the Morse-Bott index of $H_p$ for $p$ either a strongly rigid configuration, or a line configuration. We will first  focus on strongly rigid configurations, and establish the next result.

\begin{cor}
Let $G$ be a triangulated Laman graph.  Suppose each $f_{ij}$ is in $\mathcal{U}$, with $\{\bar d_{ij}|(i,j)\in E\}$ satisfying the triangle inequalities associated with $G$. Let $p$ be an equilibrium of system \eqref{MODEL}. If $p$ is strongly rigid, then $\mathcal{O}_p$ is exponentially stable. Moreover, the distance between $\vec x_i$ and $\vec x_j$ is the target distance $\bar d_{ij}$ for all $(i,j)\in E$. 
\end{cor}

\begin{proof}
Let $\{(G_i,p_i)\}^{m}_{i=1}$ be the frameworks associated with the independent partition for $(G,p)$.  Since $p$ is strongly rigid, 
 each $p_i$  consists of only two agents by Proposition \ref{indpar},  so then $m = 2N - 3$.  Also, by Proposition \ref{indpar}, each $p_i$ is an equilibrium, and hence $f_{ij}(d_{ij})=0$, which implies $d_{ij} = \bar d_{ij}$. 
 We will now compute the signature of $H_{p_i}$. Suppose $p_i$ consists of agents $\vec x_{j}$ and $\vec x_{k}$, and by Lemma \ref{LEQUIV}, we may rotate and/or translate $p$ so that both $\vec x_{j}$ and $\vec x_{k}$ are on the $x$-coordinate. 
Then $H_{p_i}$ is a $4$-by-$4$ matrix given by
\begin{equation}
H_{p_i} = 
\bar d_{jk}f'_{jk}(\bar d_{jk})
\begin{pmatrix}
-1 & 0 & 1 & 0\\
0 & 0 & 0 & 0\\
1 & 0 & -1 & 0\\
0 & 0 & 0 & 0
\end{pmatrix} 
\end{equation} 
On the other hand, we have
\begin{equation}
\frac{d}{dx}(xf_{jk}(x))\Big |_{x = \bar d_{jk}} = \bar d_{jk} f'_{jk}(\bar d_{jk}) >0 
\end{equation}
Thus, 
$n_-(H_{p_i}) = 1$. Since this holds for all $i$, we then have 
\begin{equation}
n_-(H_p) = \sum^m_{i=1}n_-(H_{p_i}) = m = 2N - 3
\end{equation} 
and hence, by the argument of dimensionality, we have  
\begin{equation}
\vec n(H_p) = (0,2N-3,3)
\end{equation} 
Thus, $\mathcal{O}_p$ is exponentially stable. This completes the proof.  
\end{proof}


\subsection{Proof of The Main Theorem}
We first focus on the case where $p\in P_{G}$ is a critical line configuration, and evaluate the signature of $H_p$. In particular, we will establish the next result.

\begin{pro}\label{LC}
Let $G$ be a triangulated Laman graph of $N$ vertices with $N> 2$. Suppose that each $f_{ij}$  is in $\mathcal{U}$, with $\{\bar d_{ij}|(i,j) \in E \}$ satisfying the triangle inequalities associated with $G$. Let $\mathcal{O}_p$ be a nondegenerate critical orbit of line configurations. Then, $n_+(H_p)>0$. 
\end{pro}

It is computationally convenient to collect the $x$-coordinates of agents $x_1$ to $x_N$ in the first $N$ entries of a vector, and the $y$-coordinates in the last $N$ entries. To this end, we
let $\vec a$ and $\vec b$ be two vectors in $\mathbb{R}^N$ containing $x$-coordinates and $y$-coordinates of agents respectively, i.e., 
\begin{equation}
\left\{
\begin{array}{ll}
\vec a := (x_1,\cdots,x_N)\\ 
\vec b := (y_1,\cdots,y_N)
\end{array}
\right.
\end{equation}
We then re-arrange entries of a configuration $p$ so that 
\begin{equation}
p = (\vec a,\vec b)
\end{equation}
By Lemma \ref{LEQUIV}, we can assume, without loss of generality, that the line configuration $p$ is aligned with the $x$-axis, or equivalently that $\vec b = 0$. An advantage of re-arranging entries is that the Hessian $H_p$ can now be expressed as a block-diagonal matrix given by
\begin{equation}\label{HESS}
H_p =
\begin{pmatrix}
A_p & 0\\
0 & B_p
\end{pmatrix}
\end{equation}
where $A_p$ and $B_p$ are $N$-by-$N$ symmetric zero-row/column-sum matrices. The $ij$-th entry, for $i\neq j$, of $A_p$ and $B_p$ are given by  
\begin{equation}\label{GMAT}
A_{ij}:=
\left\{
\begin{array}{ll}
\frac{d}{dx} (xf_{ij}(x))\big |_{x = d_{ij}} & \text{if }(i,j)\in E\vspace{3pt}\\
0 & \text{otherwise}
\end{array}
\right.
\end{equation}
and       
\begin{equation}\label{GMAT'}
B_{ij}:=
\left\{
\begin{array}{ll}
f_{ij}(d_{ij}) & \text{if }(i,j)\in E\vspace{3pt}\\
0 & \text{otherwise}
\end{array}
\right.
\end{equation}
the diagonal entries of $A_p$ and $B_p$ are then determined by the conditions that their row/column-sum are zeros. 

By Lemma \ref{LEQUIV}, the null space of $H_p$ contains  $T_p\mathcal{O}_p$, i.e., the  tangent space of $\mathcal{O}_p$ at $p$ which is the vector space spanned by the next three vectors in $\mathbb{R}^{2N}$:
\begin{equation}\label{abp}
\left\{
\begin{array}{ll}
\vec t_a: = (\vec e, 0)\\
\vec t_b: = (0,\vec e)\\
\vec r_p: = (0, \vec a)
\end{array}
\right.
\end{equation} 
where $\vec e$ is a vector of all ones in $\mathbb{R}^N$. 
The first two vectors $\vec t_a$ and $\vec t_b$ represent  infinitesimal motions of translations of $p$ along the $x$-coordinate and the $y$-coordinate, respectively. The third vector $\vec r_p$ represents an infinitesimal motion of clockwise rotation of $p$ around the origin. 
It is also straightforward to check that all three vectors  are in the null space of $H_p$. Now suppose the critical orbit $\mathcal{O}_p$ is nondegenerate; then by Lemma \ref{LEQUIV} the null space of $H_p$ should only be spanned by $\vec t_a$, $\vec t_b$ and $\vec r_p$. Further, by \eqref{HESS} and \eqref{abp}, we see that the null space of $A_{p}$ is spanned by $\vec e$, and the null space of $B_p$ is spanned by $\vec e$ and $\vec a$.

We are now in a position to prove Proposition \ref{LC}.  
\vspace{5pt}

\hspace{8pt}{\it Proof of Proposition \ref{LC}}: \hspace{4pt} We prove the proposition by showing that $n_+(B_p) >0$. The proof will be completed by induction on the number of agents. First consider the base case $N= 3$. Assume that $p$ is aligned with the $x$-coordinate with $x_2<x_1<x_3$, i.e., agent $\vec x_1$ lies in between $\vec x_2$ and $\vec x_3$.  We now show that the matrix $n_+(B_p)>0$. Since $p$ is an equilibrium, then
\begin{equation}\label{INT3}
d_{12}f_{12}(d_{12}) = d_{13} f_{13}(d_{13})= -d_{23}f(d_{23})
\end{equation} 
We now show that these three numbers are all negative. Suppose not, then we have 
\begin{equation}
\left\{
\begin{array}{l}
d_{12} \ge \bar d_{12} \\
d_{13} \ge \bar d_{13} \\
d_{23} \le \bar d_{23}
\end{array}
\right.
\end{equation} 
This holds because the function $xf_{ij}(x)$ is strictly monotonically increasing by condition C1. On the other hand, we have 
\begin{equation}
d_{12}+d_{13}=d_{23}
\end{equation}
which implies that 
\begin{equation}
\bar d_{12} + \bar d_{13}\le \bar d_{23}
\end{equation}
This then violates the triangle inequality. Thus, the three numbers in \eqref{INT3} are all negative. In particular, both $f_{12}(d_{12})$ and $f_{13}(d_{13})$ are negative. Let 
$
\vec e_1 := (1,0,0)
$ 
be a test vector. Then by computation, we have 
\begin{equation}
\langle \vec e_1,B_{p} \vec e_1\rangle  = -f_{12}(d_{12}) - f_{13}(d_{13}) > 0
\end{equation}  
Thus, $n_+(B_p) > 0$
, and hence $n_+(H_p) > 0$.

Now apply the technique of induction: We assume the fact that if $\mathcal{O}_p$ is nondegenerate, then  $n_+(B_p) > 0$ for any $N\le n$ with $n\ge 3$, and we prove for the case $N = n+1$. Fix a Henneberg construction of $G$, and without loss of generality, assume that $1$ is the last vertex joining $G$ via edges $(1,2)$ and $(1,3)$ to vertices $2$ and $3$, respectively. We still assume that $p$ is aligned with the $x$-coordinate. Then there are two cases regarding the position of agent $\vec x_1$: either $(x_1 - x_2)(x_1 - x_3)< 0$ or $(x_1- x_2)(x_1- x_3)>0$, depending on whether or not agent $\vec x_1$ lies in between $\vec x_2$ and $\vec x_3$. For simplicity, we will only focus on the former case, and assume 
$
x_2 < x_1 < x_3
$. Similar analysis can be applied to the other case as well.

Let $\vec e_1,\cdots, \vec e_{n+1}$ be the standard basis of $\mathbb{R}^{n+1}$. Similarly, we have 
\begin{equation}\label{GE1}
\langle \vec e_1, B_p\vec e_1\rangle = -f_{12}(d_{12}) - f_{13}(d_{13}) 
\end{equation}
We now show that  if $\langle \vec e_1, B_p\vec e_1\rangle \ge 0$, then  $n_+(B_p)>0$. Since $\mathcal{O}_p$ is nondegenerate, the null space of $B_p$ is spanned by $\vec e$ and $\vec a$ only. On the other hand, the three vectors $\vec e_1$, $\vec e$ and $\vec a$ are linearly independent, so $B_p \vec e_1\neq 0$. Let $\lambda_1,\cdots, \lambda_{n-1}$ be the non-zero eigenvalues of $B_p$, and let $\vec v_i$ be the unit-length eigenvector of $B_p$ corresponding to $\lambda_i$, then 
\begin{equation}
\langle \vec e_1, B_p\vec e_1\rangle = \sum^{n-1}_{i=1} \lambda_i \langle\vec e_1, \vec v_i\rangle^2 \ge 0
\end{equation}
Since there exists some $i$ with $\langle \vec e_1,\vec v_i\rangle\neq 0$, there must exist at least one positive eigenvalue of $B_p$.    

So in the rest of the proof, we only consider the case $\langle \vec e_1, B_p\vec e_1\rangle <0$. Since $\vec x_1$ is balanced in $p$, we have
\begin{equation} 
d_{12} f_{12}(d_{12}) = d_{13}f_{13}(d_{13})
\end{equation}
Then by expression \eqref{GE1}, both $f_{12}(d_{12})$ and $f_{13}(d_{13})$ are positive. In particular, we have 
\begin{equation}
d_{23} = d_{12} + d_{13} > \bar d_{12} + \bar d_{13} > \bar d_{23}
\end{equation}
Now choose a function $g\in C^1(\R_+,\R)$ such that it satisfies the next three conditions 
\begin{itemize}
\item[1.] $g$ satisfies condition C1, and $g(\bar d_{23}) = 0$ 
\vspace{2pt}
\item[2.] $d_{23}g(d_{23}) = d_{12} f_{12}(d_{12}) = d_{13}f_{13}(d_{13})$
\vspace{2pt}
\item[3.] $\frac{d}{dx}(xg(x))\big |_{x = d_{23}} = A_{12} A_{13}/(A_{12} + A_{13})$
\end{itemize} 
\vspace{2pt} 
with $A_{ij}$ the $ij$-th entry of $A_p$  defined in \eqref{GMAT}.

We introduce function $g$ because of the following fact: Let $G' = (V',E')$ be the subgraph of $G$ induced by vertices $V':=\{2,\cdots,n+1\}$, and let $(G',p')$ be the corresponding framework. If we replace $f_{23}$ with 
\begin{equation}
\tilde f_{23} := f_{23} + g,
\end{equation}
then $p'$ is an equilibrium of the sub-system induced by $G'$, with the modification that $f_{23}$ is replaced by $\tilde f_{23}$. To see this, it suffices to check that agents $\vec x_2$ and $\vec x_3$ in $p'$ are still balanced. But this holds because of the second condition on $g$. We note that the first condition on $g$ implies that $\tilde f_{23}\in \mathcal{U}$ with $\tilde f_{23}(d_{23}) = 0$. The third condition is a technical condition, and will be justified later.  Also note that $G'$ is a triangulated Laman graph, and $\{\bar d_{ij}| (i,j)\in E'\}$ satisfies the triangle inequalities associated with $G'$. Thus, we can apply the technique of induction on the critical orbit $\mathcal{O}_{p'}$ of the modified sub-system.

Let $\Phi'$ be the potential function associated with the modified sub-system induced by $G'$. Let $H_{p'}$ be the Hessian of $\Phi'$ at $p'$. Similarly, we can express $H_{p'}$ as a block-diagonal matrix as
\begin{equation}
H_{p'} = 
\begin{pmatrix}
A_{p'} & 0\\
0 & B_{p'}
\end{pmatrix}
\end{equation}
with $A_{p'}$ and $B_{p'}$ defined in the same way as $A_{p}$ and $B_{p}$ but with respect to $G'$. Also we replace $f_{23}$ and $f'_{23}$ with $\tilde f_{23}$ and $\tilde f'_{23}$, respectively. We will now introduce a formula that relates the signature of $H_p$ to the signature of $H_{p'}$. First we introduce a vector-valued sign function as
\begin{equation}
sgn(x):=\left\{
\begin{array}{ll}
(1,0,0) & \text{ if } x > 0\\
(0,1,0) & \text{ if } x < 0\\
(0,0,1) & \text{ if } x = 0
\end{array}\right.
\end{equation} 
and recall that $\vec n(H) = (n_+(H),n_-(H),n_0(H))$ is defined as the signature of $H$. Then, 
\begin{equation}\label{F1}
\left\{
\begin{array}{l}
\vec n(A_p) = \vec n(A_{p'})
+ sgn(-A_{12} - A_{13}) \vspace{3pt}\\
\vec n(B_p) = \vec n(B_{p'})
+sgn(-B_{12}-B_{13})
\end{array}\right.
\end{equation} 
where $A_{ij}$ and $B_{ij}$ are entries of $A_p$ and $B_{p}$, respectively (the validity of this formula requires the second and the third conditions on $g$). The proof of the formula is provided in the Appendix.

From \eqref{F1}, we see that if $\mathcal{O}_p$ is nondegenerate in the original system, then so is $\mathcal{O}_{p'}$ in the modified sub-system. Thus, by induction we have $n_+(B_{p'})>0$. Then applying \eqref{F1} again, we conclude that $n_+(B_p)>0$. This then completes the proof. \hfill{\QED}
\vspace{5pt}

With the results above, we will now return to proof Theorem \ref{MAIN}. Let $p$ be an equilibrium of system \eqref{MODEL}. If $\mathcal{O}_p$ is strongly rigid, then $\mathcal{O}_p$ is (exponentially) stable as we have shown at the end of section~\ref{ssec:morsebott}. So we assume now that $\mathcal{O}_p$ is not strongly rigid, and we show that $\mathcal{O}_p$ is unstable.

Let $p_1,\cdots,p_m$ be the line sub-configurations of $p$ associated with the independent partition, and without loss of generality, we assume that  $p_1$ contains at least three agents. Since $\mathcal{O}_p$ is a nondegenerate critical orbit, then so is $\mathcal{O}_{p_1}$ by Corollary \ref{cor6}, and hence the co-index $n_+(H_{p_1})$ must be positive by Proposition \ref{LC}. We then apply the Morse-Bott index formula, i.e.,  
\begin{equation}
n_+(H_p) = \sum^m_{i=1}n_+(H_{p_i})
\end{equation} 
to conclude that the Hessian matrix $H_p$ also has at least one positive eigenvalue. So we have shown that  a critical orbit is stable if and only if it is strongly rigid. The set of stable critical orbits is characterized by the condition that $d_{ij}=\bar d_{ij}$ for all $(i,j)\in E$, and hence there are as many as $2^{N-2}$ stable critical orbits  in total. The convergence of system \eqref{MODEL} is implied by Lemma \ref{CONVGG}.

\section{Conclusions}
Design of control laws that only stabilize the target configurations of a formation is known to be a challenging problem. Indeed, the conjunction of the decentralization constraints and the nonlinear nature of the dynamics lead to the appearance of undesired equilibria in the system. Counting these equilibria in general is a difficult and open problem, let alone characterizing them. In this paper,  we have provided a partial solution by exhibiting a class of undirected graphs and control laws for which only desired configurations are stable.  We have furthermore derived results characterizing the extremal points of a class of  equivariant Morse functions that might be of an independent interest. 

 \bibliographystyle{unsrt}
\bibliography{FC}

\section*{Appendix}

\subsection{Sketch of the proof of Proposition \ref{MBIF}}
The Hessian matrix $H_p$ considered here is with respect to the  arrangement $p = (\vec x_1,\cdots,\vec x_N)$. Let $H_i$ be a $2N$-by-$2N$ symmetric matrix derived by adding zero rows and columns to $H_{p_i}$. The $(2j-1)$-th and $2j$-th rows/columns of $H_i$  are zero rows/columns if $j$ is not a vertex of $G_i$, and if we remove these zero rows and columns, then we recover $H_{p_i}$. It should be clear that $n_+(H_i) = n_+(H_{p_i})$ and $n_-(H_i) = n_-(H_{p_i})$. 
We then express $H_p$ as 
\begin{equation}
H_p = \sum^m_{i=1} H_i
\end{equation}     
It now suffices to show that 
\begin{equation}\label{INDEXF}
\left\{
\begin{array}{l}
n_-(H_p)=\sum^m_{i=1}n_-(H_{i}) \vspace{3pt}\\
n_+(H_p)=\sum^m_{i=1}n_+(H_{i})
\end{array}
\right.
\end{equation}
Each $H_{p_i}$ has at least three zero eigenvalues. Let $\lambda_{i_1},\cdots, \lambda_{i_{l_i}}$ be the other eigenvalues of $H_{p_i}$, and for simplicity, assume that they are all nonzero. It should be clear that $l_i = |E_i|$, and hence $\sum^m_{i=1} l_i = 2N - 3$.  Suppose for the moment that for each $\lambda_{i_j}$, we can find a vector $u_{i_j}\in \mathbb{R}^{2N}$ so that the ensemble of these vectors satisfies the following condition:
\begin{equation}\label{ortho}
\langle u_{j'_{k'}},  H_{i} u_{j_k}\rangle = \delta_{ij}\delta_{jj'}\delta_{kk'}\lambda_{j_k}  
\end{equation} 
with $\delta$ the Kronecker delta. Then we can define a  matrix $U$ with its column vectors $u_{i_j}$'s such that 
\begin{equation}
U^{\top} H_p U = \operatorname{diag}(\Lambda_1,\cdots, \Lambda_m)
\end{equation}
with $\Lambda_i:= \operatorname{diag}(\lambda_{i_1},\cdots, \lambda_{i_{l_i}})$. Thus, Proposition \ref{MBIF} immediately follows from the Sylvester's Law of inertia.

We will now describe how we construct the vector $u_{i_j}$. First consider a simple example: Suppose we have a nondegenerate triangle $\vec x_1$, $\vec x_2$ and $\vec x_3$ on the plane, then for sufficiently small perturbation $\delta \vec x_i$ of agent $\vec x_i$ for $i=1,2$, we can find a unique displacement $\delta \vec x_3$ of $\vec x_3$ such that  we can maintain the distances $d_{12}$ and $d_{23}$ by following this displacement, i.e., 
\begin{equation}
\|\vec x_3 - \vec x_i\| = \|(\vec x_3 + \delta \vec x_3) - (\vec x_i + \delta \vec x_i)\|, \hspace{10pt} \forall i =1,2 
\end{equation}
In fact, if we let $\rho$ be the map 
\begin{equation}
\rho: (\delta \vec x_1,\delta \vec x_2) \mapsto \delta \vec x_3 
\end{equation}
then by the inverse function theorem, the map $\rho$ is well-defined over a small neighborhood of the origin in $\R^4$. Moreover, $\rho$ is smooth and $\rho(0) = 0$. Thus, we can consider the derivative map 
\begin{equation}
d\rho_0: \R^4 \to \R^2
\end{equation}
at the origin, which describes the infinitesimal motion of the displacement of $\vec x_3$ with respect to the infinitesimal motions of perturbations of $\vec x_1$ and $\vec x_2$. 
This geometric fact can be generalized to an arbitrary framework $(G,p)$ with $G$ a triangulated Laman graph. Precisely, we let $\{(G_i,p_i)\}^m_{i=1}$ be the frameworks associated with the independent partition for $(G,p)$. Then  we can perturb one sub-configuration $p_i$ while preserving the shapes of the others \cite{chenRMAP1}, i.e., if we let $\delta p_i$ be the perturbation of $p_i$, then there is a unique displacement $\delta p_{-i}$ for the rest agents  $p_{-i}$ such  that $p_{-i} + \delta p_{-i}$ can be derived by rotating and/or translating of $p_i$ in $\R^2$. The map
\begin{equation}
\rho_i:\delta p_i \mapsto \delta p_{-i}
\end{equation}
is well defined over a small neighborhood of the origin, and similarly $\rho$ is smooth and $\rho(0) = 0$. Thus, we can still consider the derivative map $d\rho_0$ which describes the infinitesimal version of the displacement  of $p_{-i}$.

We now return to construction of the vector $u_{i_j}$. Fix an $i$, and let $v_{i_j}$ be the unit-length eigenvector of $H_{p_i}$ corresponding to eigenvalue $\lambda_{i_j}$.  We now treat $v_{i_j}$ as the infinitesimal version of the perturbation of $p_i$, and correspondingly we let
\begin{equation}
w_{i_j}:= d\rho_0(v_{i_j})
\end{equation} 
be the infinitesimal version of the displacement of $p_{-i}$.  
For simplicity, we assume that $p_i$ consists of the first $k$ agents, then we construct $u_{i_j}$ by concatenating $v_{i_j}$ and $w_{i_j}$ as
\begin{equation}
u_{i_j} := (v_{i_j},w_{i_j})
\end{equation}
We then show in \cite{chenRMAP1} that the ensemble of the vectors $u_{i_j}$ satisfies the desired condition described by \eqref{ortho}.  

\subsection{Proof of formula \eqref{F1}}
Let $\vec v_1,\cdots, \vec v_n\in\R^n$ be the unit-length eigenvectors of $A_{p'}$ corresponding to eigenvalues $\lambda_1,\cdots, \lambda_n$. We now define, for each $\vec v_i$, a vector $\vec v^*_i\in \R^{n+1}$ as follows. Let $v_{ij}$ be the $j$-th entry of $\vec v_i$; then
\begin{equation}
\alpha_i :=  \frac{A_{12}v_{i1} + A_{13} v_{i2}}{A_{12} + A_{13}} 
\end{equation} 
Note that this is well defined because by condition C1, $A_{12}$ and $A_{13}$ are always positive. Now let 
\begin{equation}
\vec v^*_{i} := (\alpha_i, \vec v_i)  
\end{equation}
Then by using the third condition on $g$, we can get 
\begin{equation}
A_{p} \vec v^*_i = \lambda_i (0,\vec v_i)
\end{equation}
Now let $Q_A:= (\vec e_1,\vec v^*_1,\cdots, \vec v^*_n)$ be an $(n+1)$-by-$(n+1)$ matrix; then 
\begin{equation}
Q^{\top}_A A_{p} Q_A= \operatorname{diag}(-A_{12}-A_{13}, \lambda_1,\cdots, \lambda_n) 
\end{equation}
By Sylvester's Law of inertia, we have 
\begin{equation}
\vec n(A_p) = \vec n(A_{p'}) + sgn(-A_{12} - A_{13})
\end{equation}
The analysis for the other part is similar. 
Let $\vec u_1,\cdots, \vec u_n\in\R^n$ be the unit-length eigenvectors of $B_{p'}$ corresponding to eigenvalues $\mu_1,\cdots, \mu_n$. For each $\vec u_i$, we let
\begin{equation}
\beta_i :=  \frac{(x_3- x_1)u_{i1} + (x_1 - x_2)u_{i2} }{x_3 - x_2} 
\end{equation} 
This is also well defined because $\vec x_2$ and $\vec x_3$ are  on the $x$-coordinate, but at different positions.  Now let 
\begin{equation}
\vec u^*_{i} := (\beta_i, \vec u_i)  
\end{equation}
Then by using the second condition on $g$, we can get 
\begin{equation}
B_{p} \vec u^*_i = \mu_i (0,\vec u_i)
\end{equation}
Letting $Q_B:= (\vec e_1,\vec u^*_1,\cdots, \vec u^*_n)$, it then follows that 
\begin{equation}
Q^{\top}_B B_{p} Q_B= \operatorname{diag}(-B_{12}-B_{13}, \mu_1,\cdots, \mu_n) 
\end{equation}
This then shows that
\begin{equation}
\vec n(B_p) = \vec n(B_{p'}) + sgn(-B_{12} - B_{13})
\end{equation}
which completes the proof.

\end{document}